%% file: camera_ready.tex
\documentclass[a4paper,UKenglish,cleveref, autoref, thm-restate]{lipics-v2021}

\pdfoutput=1 
\hideLIPIcs  

\title{A Dichotomy Theorem for Linear Time Homomorphism Orbit Counting in Bounded
	Degeneracy Graphs} 

\titlerunning{Homomorphism Orbit Counting in Bounded
	Degeneracy Graphs} 

\author{Daniel Paul-Pena}{University of California, Santa Cruz, United States }{dpaulpen@ucsc.edu}{https://orcid.org/0009-0008-1073-6173}{}

\author{C. Seshadhri}{University of California, Santa Cruz, United States}{sesh@ucsc.edu}{https://orcid.org/0000-0003-2163-3555}{}

\authorrunning{D. Paul-Pena and C. Seshadhri} 

\Copyright{Daniel Paul-Pena and C. Seshadhri} 



\category{} 

\relatedversion{} 

\funding{Both authors are supported by NSF CCF-1740850, CCF-1839317, CCF-2402572, and DMS-2023495
}

\nolinenumbers 

\input{sesh-macros.tex}

\usepackage{algorithm}
\usepackage{algpseudocode}
\usepackage{amsmath}
\usepackage{comment}
\usepackage{tikz}
\usepackage{caption}
\usepackage{mathtools,amsfonts}

\newcommand{\Hom}[2]{\mathrm{Hom}_{#2}(#1)}
\newcommand{\HomAl}[1]{\mathrm{Hom}_{#1}}
\newcommand{\Aut}{\mathrm{Aut}}
\newcommand{\ext}{\mathrm{ext}}
\newcommand{\down}{\mathrm{down}}

\newcommand{\OHC}{\mathrm{OrbitHom}}
\newcommand{\VHC}{\mathrm{VertexHom}}
\newcommand{\ComputeVHC}{ComputeVertexCounts}
\newcommand{\HS}{H_S}
\newcommand{\Signature}{\text{Sig}}
\newcommand{\Reachable}{Reach}

\newcommand{\IS}{\mathcal{I}\mathcal{S}}

\newcommand{\Agg}{Agg}

\newcommand{\dirG}{G^\to}

\begin{document}

\maketitle

\begin{abstract}
Counting the number of homomorphisms of a pattern graph $H$ in a large input graph $G$ is a fundamental problem in computer science. 
In many applications in databases, bioinformatics, and network science, we need more than just the total count. 
We wish to compute, for each vertex $v$ of $G$, the number of $H$-homomorphisms that $v$ participates in. This problem
is referred to as \emph{homomorphism orbit counting}, as it relates to the orbits of vertices of $H$ under its automorphisms.

Given the need for fast algorithms for this problem, we study when near-linear time algorithms are possible. A natural restriction is to
assume that the input graph $G$ has bounded degeneracy, a commonly observed property in modern massive networks. Can we characterize
the patterns $H$ for which homomorphism orbit counting can be done in near-linear time?

We discover a dichotomy theorem that resolves this problem. For pattern $H$, let $\ell$ be the length of the longest
induced path between any two vertices of the same orbit (under the automorphisms of $H$).
If $\ell \leq 5$, then $H$-homomorphism orbit counting can be done in near-linear time for
bounded degeneracy graphs. If $\ell > 5$, then (assuming fine-grained complexity conjectures)
there is no near-linear time algorithm for this problem. We build on existing work on dichotomy theorems
for counting the total $H$-homomorphism count.
Surprisingly, there exist (and we characterize)
patterns $H$ for which the total homomorphism count can be computed in near-linear time, but the
corresponding orbit counting problem cannot be done in near-linear time.
\end{abstract}

\section{Introduction} \label{sec:intro}

Analyzing the occurrences of a small pattern graph $H$ in a large input graph $G$
is a central problem in computer science. The theoretical study has led to
a rich and immensely deep theory~\cite{Lo67,ChNi85,FlGr04,DaJo04,Lo12,AhNeRo+15,CuDeMa17,PiSeVi17,RoWe20, BrRo22, BrGoMe+23}.
The applications of graph pattern counts occur across numerous scientific areas,
including logic, biology, statistical physics, database theory, social sciences, machine learning, and network science~\cite{HoLe70,ChMe77,Co88,BrWi99,DrGr00,BoChLo+06,Fa07,OcHaCa+12,UgBaKl13,PiSeVi17,DeRoWe19,PaSe20}. (Refer to the tutorial~\cite{SeTi19} for more details on applications.)

A common formalism used for graph pattern counting is \emph{homomorphism counting}. The pattern graph is denoted $H = (V(H), E(H))$ and is assumed
to have constant size. The input graph is denoted $G = (V(G), E(G))$. Both graphs are simple and do not contain self-loops. An $H$-homomorphism is a map $f:V(H) \to V(G)$ that preserves edges. Formally, $\forall (u,v) \in E(H)$, $(f(u), f(v)) \in E(G)$. Let $\Hom{G}{H}$ denote the number of distinct $H$-homomorphisms in $G$.

Given the importance of graph homomorphism counts, the study of efficient algorithms for this problem
is a subfield in itself~\cite{ItRo78,AlYuZw97,BrWi99,DrGr00,DiSeTh02,DaJo04,BoChLo+06,CuDeMa17,Br19,RoWe20}. The simplest version of this problem is when $H$ is a triangle, itself a problem that attracts much attention.
Let $n = |V(G)|$ and $k = |V(H)|$. 
Computing $\Hom{G}{H}$ is $\#W[1]$-hard when parameterized by $k$ (even when $H$ is a $k$-clique), so we do not expect $n^{o(k)}$ algorithms for general $H$~\cite{DaJo04}. 
Much of the algorithmic study of homomorphism counting is in understanding conditions on $H$ and $G$ when the trivial $n^k$ running time bound can be beaten.

Our work is inspired by the challenges of modern applications of homomorphism counting, especially in network science.
Typically, $n$ is extremely large, and only near-linear time ($n\cdot \poly(\log{n})$) algorithms  are feasible.
Inspired by a long history and recent theory on this topic,
we focus on \emph{bounded degeneracy} input graphs (we say bounded degeneracy graphs to refer to graphs belonging to classes of graphs with bounded degeneracy). This includes all non-trivial minor-closed
graph families, such as planar graphs, bounded genus graphs, and bounded tree-width graphs.
Many practical algorithms for large-scale graph pattern counting use algorithms for bounded
degeneracy graphs~\cite{AhNeRo+15,JhSePi15,PiSeVi17,OrBr17,JaSe17,PaSe20}. 
Real-world graphs typically have a small degeneracy, comparable to their average degree (\cite{GoGu06,JaSe17,ShElFa18,BeChGh20,BeSe20}, also Table 2 in~\cite{BeChGh20}). 

Secondly, many modern applications for homomorphism counting
require more fine-grained statistics than just the global count $\Hom{G}{H}$. 
The aim is to find, \emph{for every vertex $v$ of $G$},
the number of homomorphisms that $v$ participates in. 
Seminal work in network analysis for bioinformatics plots the
distributions of these per-vertex counts to compare graphs~\cite{ItLeKa+05,Pr07}. 
Orbit counts can be used to generate features for vertices,
sometimes called the graphlet kernel~\cite{ShViPe+09}. 
In the past few years, there have been many applications of these per-vertex counts~\cite{BeHeLa+11,UgBaKl13,SaSePi+15,Ts15,BeGlLe16,TsPaMi17,RoKaKl+17,YiBiLe18,YiBiLe19}.

Algorithms for this problem require considering the ``roles'' that $v$ could play in a homomorphism.
For example, in a $7$-path (a path of length $6$) there are $4$ different roles: a vertex $v$ could be in the middle, 
could be at the end, or at two other positions. These roles are colored in \Fig{graphs_loipl}. 
The roles are called \emph{orbits} (defined in the Section \ref{sec:preliminaries}), and the problem 
of \emph{$H$-homomorphism orbit counting} is as follows:
for every orbit $\psi$ in $H$ and every vertex $v$ in $G$, output the number of homomorphisms
of $H$ where $v$ participates in the orbit $\psi$.
This is the main question addressed by our work:

\medskip

\emph{What are the pattern graphs $H$ for which the $H$-homomorphism orbit counting
	problem is computable in near-linear time (when $G$ has bounded degeneracy)?}

\medskip

Recent work of Bressan followed by Bera-Pashanasangi-Seshadhri introduced the question
of homomorphism counting for bounded degeneracy graphs, from a fine-grained complexity perspective \cite{Br19,BePaSe21}. 
A dichotomy theorem for near-linear time counting of $\Hom{G}{H}$ was provided in subsequent work \cite{BeGiLe+22}. Assuming fine-grained
complexity conjectures, $\Hom{G}{H}$ can be computed in near-linear time iff the longest
induced cycle of $H$ has length at most $5$. It is natural to ask whether
these results extend to orbit counting.

\begin{figure}
	\centering
	\includegraphics[width=\linewidth*2/3]{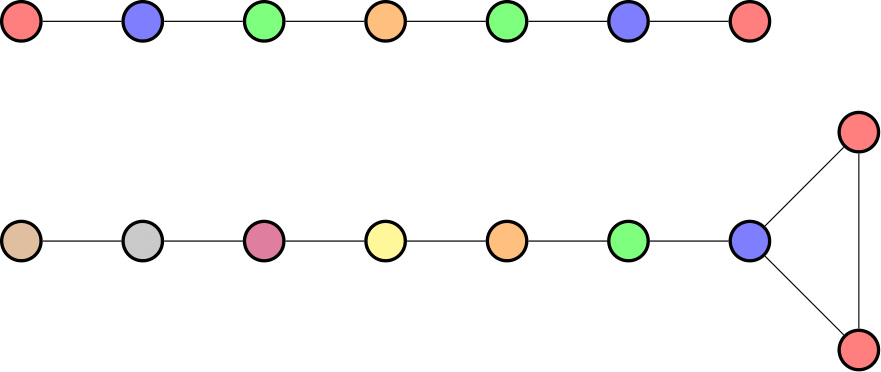}
	\caption{Examples of orbits and LIPCO values. Vertices in the same orbit have the same color. 
		The top graph is the 7-path (a path of length 6). There is an induced path of length 6 between the red vertices, hence the \LIPCO{} of this graph is 6.
		\Thm{main} implies that we can not compute $\OHC$ in near-linear time.
		\\The bottom graph adds a triangle at the end, breaking the symmetry, and the only vertices in the same orbit in that graph are the red ones. The \LIPCO{} in this graph is now less than 6 so we can compute $\OHC$
		in near-linear time.}
	\label{fig:graphs_loipl}
\end{figure}

\subsection{Main Result} \label{sec:result}

We begin with some preliminaries.
The input graph $G = (V(G), E(G))$ has $n$ vertices and $m$ edges. 
A central notion in our work is that of graph degeneracy, also called
the coloring number.

\begin{definition} \label{def:degen} A graph $G$ is $\degen$-degenerate if the minimum degree in every subgraph of $G$
	is at most $\degen$.
	
	The degeneracy of $G$ is the minimum value of $\degen$
	such that $G$ is $\degen$-degenerate.
\end{definition}

A family of graphs has \emph{bounded
	degeneracy} if the degeneracy is constant with respect to the graph size. 
Bounded degeneracy graph classes are extremely rich. For example,
all non-trivial minor-closed families have bounded degeneracy. This includes
bounded treewidth graphs. Preferential
attachment graphs also have bounded degeneracy; real-world graphs
have a small value of degeneracy (often in the 10s) with respect
to their size (often in the hundreds of millions)~\cite{BeChGh20}.

We assume the pattern graph $H=(V(H), E(H))$ to have
a constant number of vertices. (So we suppress any dependencies on purely the size of $H$.) Consider the group of automorphisms
of $H$. The vertices of $H$ can be partitioned into \emph{orbits},
which consist of vertices that can be mapped to each other by some 
automorphism (defined formally in Definition \ref{def:orbits}). For example, in \Fig{graphs_loipl}, the $7$-path has four different
orbits, where each orbit has the same color. The $7$-path with a hanging triangle (in \Fig{graphs_loipl})
has more orbits, since the pattern is no longer symmetric with respect to the ``center'' of the $7$-path and hence the opposite ``ends'' of the $7$-path cannot be mapped
by a non-trivial automorphism.

The set of orbits of the pattern $H$ is denoted $\Psi(H)$. Let $\Phi(H,G)$ be the set of homomorphisms from $H$ to $G$ ($\Hom{G}{H}=|\Phi(H,G)|$). We now define our main problem.

\begin{definition} \label{def:ohc}
	{Homomorphism Orbit Counts}: For each orbit $\psi \in \Psi(H)$ and vertex $v \in V(G)$, 
	define $\OHC_{H,\psi}(v)$ to be the number of $H$-homomorphisms mapping a vertex of $\psi$ to $v$.
	Formally, $\OHC_{H,\psi}(v) = |\{\phi \in \Phi(H,G) \colon \exists h\in \psi, \, \phi(h)=v\}|$.
	
	The problem of \emph{$H$-homomorphism orbit counting} is to output
	the values $\OHC_{H,\psi}(v)$ for all $v \in V(G), \psi \in \Psi(H)$. (Abusing
	notation, ${\OHC}_{{H}}({G})$ refers to the list/vector of all of these values.)
\end{definition}

Note that for a given $H$, the size of the output is $n|\Psi(H)|$ (recall $n = |V(G)|$). For example,
when $H$ is the $7$-path, we will get $4n$ counts, for each vertex and each of the four orbits.

Our main result is a dichotomy theorem that precisely characterizes patterns $H$
for which $\OHC_{H}(G)$ can be computed in near-linear time. 
We introduce a key definition.

\begin{definition} \label{def:lipo} For a pattern $H$, the \emph{Longest Induced Path
		Connecting Orbits} of $H$, denoted $\LIPCO(H)$ is defined as follows.
	It is the length of the longest induced simple path, measured in edges, between any two vertices $h, h'$ in $H$ 
	(where $h$ may be equal to $h'$, forming a cycle) in the same orbit.
\end{definition}

Again refer to \Fig{graphs_loipl}. The $7$-path has a \LIPCO{} of six, since
the ends are in the same orbit. On the other hand, the second pattern
($7$-path with a triangle) has a \LIPCO{} of $3$ due to the triangle.

The Triangle Detection Conjecture was introduced by Abboud and Williams on the complexity of determining whether a graph has a triangle~\cite{AbWi14}. 
It is believed that this problem cannot be solved in near-linear time, and indeed,
may even require $\Omega(m^{4/3})$ time. We use this conjecture for the lower bound of our main theorem.

\begin{conjecture}[Triangle Detection Conjecture~\cite{AbWi14}]
	\label{conj:triangle}
	There exists a constant $\gamma>0$ such that in the word RAM model
	of $O(\log n)$ bits, any algorithm to detect whether an input graph on
	$m$ edges has a triangle requires $\Omega(m^{1+\gamma})$ time
	in expectation.
\end{conjecture}

Our main theorem proves that the \LIPCO{} determines the dichotomy. Note that because $G$ is a bounded degeneracy graph we have $m=O(n)$, we will be expressing the bounds in terms of $m$.

\begin{theorem} \label{thm:main}
	{(Main Theorem)} Let $G$ be a graph with $n$ vertices, $m$ edges, and bounded degeneracy. Let $\gamma > 0$ denote the constant from the Triangle Detection Conjecture (\Conj{triangle}).
	\begin{itemize}
		\item If $\LIPCO(H) \leq 5$: there exists a deterministic algorithm that computes $\OHC_{H}(G)$ in time $O(m\log n)$.\footnote{The exact dependency on the degeneracy $\degen$ of the input graph $G$ is $O\left(\kappa^{|H|-1}\right)$.}
		\item If $\LIPCO(H) > 5$: assume the Triangle Detection Conjecture. There is no algorithm with (expected) running time $O(m^{1+\gamma})$ that computes $\OHC_{H}(G)$.
	\end{itemize}
\end{theorem}

\subsubsection*{Orbit Counting vs Total Homomorphism Counting} In the following discussion, we use ``linear'' to actually mean
near-linear, we assume that the Triangle Detection Conjecture is true, and we assume that $G$ has bounded degeneracy.

One of the most intriguing aspects of the dichotomy of \Thm{main} is that it differs
from the condition for getting the total homomorphism count. As mentioned earlier, the inspiration for \Thm{main}
is the analogous result for determining $\Hom{G}{H}$. There is a near-linear time algorithm iff
the length of the longest induced cycle (LICL) of $H$ is at most five. Since the definition of \LIPCO{} 
considers induced cycles (induced path between a vertex to itself), if $\LIPCO(H) \leq 5$, then $LICL(H) \leq 5$.
This implies, not surprisingly, that the total homomorphism counting problem is easier than the orbit counting problem.

But there exist patterns $H$ for which the orbit counting problem is provably harder than total homomorphism counting, a simple example is the $7$-path (path with $7$ vertices). There is a simple
linear time dynamic program for counting the homomorphism of paths.
But the endpoints are in same orbit, so the \LIPCO{} is six, and \Thm{main} proves the non-existence of linear time algorithms
for orbit counting. On the other hand, the \LIPCO{} of the $6$-path is five, so orbit counting can be done in linear time.

Consider the pattern at the bottom of \Fig{graphs_loipl}.
The LICL is three, so the total homomorphism count
can be determined in linear time. Because the ends of the underlying $7$-path lie in different orbits,
the \LIPCO{} is also three (by the triangle). \Thm{main} provides a linear time algorithm for orbit counting.

\subsection{Main Ideas}

The starting point for homomorphism counting on bounded degeneracy graphs
is the seminal work of Chiba-Nishizeki on  using acyclic graph orientations~\cite{ChNi85}. It is known
that, in linear time, the edges of a bounded degeneracy graph can be acyclically oriented while keeping the outdegree bounded~\cite{MaBe83}.
For clique counting, we can now use a brute force algorithm in all out neighborhoods, and get a linear time algorithm.
Over the past decade, various researchers observed that this technique can generalize to certain other pattern
graphs~\cite{Co09,PiSeVi17,OrBr17,PaSe20}. Given a pattern $H$, one can add the homomorphism counts of all acyclic orientations of $H$ for an acyclic orientation of $G$.
In certain circumstances, each acyclic orientation can be efficiently counted by a carefully tailored
dynamic program that breaks the oriented $H$ into subgraphs spanned by rooted, directed trees.

Bressan gave a unified treatment of this approach through the notion of \emph{DAG-tree decompositions}. \cite{Br19}
These decompositions give a systematic way of breaking up an oriented pattern into smaller pieces,
such that homomorphism counts can be computed by a dynamic program. Bera et al. showed that if the LICL of $H$
is at most $5$, then the DAG-treewidth of $H$ is at most one~\cite{BePaSe21,BeGiLe+22}. This immediately implies
Bressan's algorithm runs in linear time. 

Our result on orbit counting digs deeper into the mechanics of Bressan's algorithm. To run in linear time, Bressan's algorithm
requires ``compressed'' data structures that store information about homomorphism counts.
For example, the DAG-tree decomposition based algorithm can count $4$-cycles in linear time for bounded degeneracy graphs
(this was known from Chiba-Nishizeki as well~\cite{ChNi85}). But there could exist quadratically many 
$4$-cycles in such a graph. Consider two vertices connected by $\Theta(n)$ disjoint paths of length $2$;
each pair of paths yields a distinct $4$-cycle. Any linear time algorithm for $4$-cycle counting has to carefully index
directed paths and combine these counts, without actually touching every $4$-cycle.

By carefully looking at Bressan's algorithm, we discover that ``local'' per-vertex information
about $H$-homomorphisms can be computed. Using the DAG-tree decomposition, one can combine
these counts into a quantity that looks like orbit counts. Unfortunately, we cannot get
exact orbit counts, but rather a weighted sum of homomorphisms. 

To extract exact orbit counts, we dig deeper into the relationship between orbit counts
and per-vertex homomorphism counts. This requires looking into the behavior of independent
sets in the orbits of $H$. We then design an inclusion-exclusion formula that
``inverts'' the per-vertex homomorpishm counts into orbit counts. 
The formula requires orbit counts for other patterns $H'$ that are constructed by merging
independent sets in the same orbit of $H$.

Based on previous results, we can prove that if the LICL of all these $H'$ patterns
is at most $5$, then $\OHC_H(G)$ can be computed in (near)linear time. This LICL condition
over all $H'$ is equivalent to the LIPCO of $H$ being at most $5$. Achieving the upper
bound of \Thm{main}.

The above seemingly ad hoc algorithm optimally characterizes
when orbit counting is linear time computable. To prove the matching lower bound,
we use tools from the breakthrough work of Curticapean-Dell-Marx~\cite{CuDeMa17}.
They prove that the complexity of counting linear combinations of homomorphism counts
is determined by the hardest individual count (up to polynomial factors). Gishboliner-Levanzov-Shapira give
a version of this tool for proving linear time hardness~\cite{GiLeSh20}. Consider a pattern $H$ with LIPCO at least six.
We can construct a pattern $H'$ with LICL at least six by merging vertices of an orbit in $H$.
We use the tools above to construct a constant number of linear sized graphs $G_1, G_2, \ldots, G_k$
such that a linear combination of $H$-orbit counts on these graphs yields
the total $H'$-homomorphism count on $G$. The latter problem is hard by existing bounds,
and hence the hardness bounds translate to $H$-orbit homomorphism counting.

\section{Related Work} \label{sec:related}

Homomorphism and subgraph counting on graphs is an immense topic with an extensive
literature in theory and practice. For a detailed discussion of practical applications, we refer the reader
to a tutorial~\cite{SeTi19}. 

Homomorphism counting is intimately connected with the treewidth of the pattern $H$.
The notion of tree decomposition and treewidth were introduced in a seminal work by Robertson and Seymour~\cite{RoSe83,RoSe84,RoSe86}; although it has been discovered before under
different names~\cite{BeBr73,Ha76}.
A classic result of Dalmau and Jonsson~\cite{DaJo04} proved
that $\Hom{G}{H}$ is polynomial time solvable if and
only if $H$ has bounded treewidth, otherwise it is
$\#W[1]$-complete. 
D{\'i}az~et al~\cite{DiSeTh02} gave an algorithm for homomorphism counting
with runtime $O(2^{k}n^{\tw(H)+1})$ where $\tw(H)$
is the treewidth of the pattern graph $H$ and $k$ the number of vertices of $H$.

To improve on these bounds, recent work has focused on restrictions on the input $G$~\cite{RoWe20}.
A natural restriction is bounded degeneracy, which is a nuanced measure
of sparsity introduced by early work of Szekeres-Wilf~\cite{SzWi68}.
Many algorithmic results exploit low degeneracy for faster subgraph counting problems~\cite{ChNi85,Ep94,AhNeRo+15,JhSePi15,PiSeVi17,OrBr17,JaSe17,PaSe20}.

Pioneering work of Bressan introduced the concept of DAG-treewidth for faster algorithms
for homomorphism counting in bounded degeneracy graphs~\cite{Br19}. Bressan gave an algorithm for counting $\Hom{G}{H}$ running in time
essentially $m^{\dtw(H)}$, where $\dtw$ denotes the DAG-treewidth. The result
also proves that (assuming ETH) there is no algorithm running in time $m^{o(\dtw(H)/\log \dtw(H))}$. 

Bera-Pashanasangi-Seshadhri build on Bressan's methods to discover a dichotomy theorem for linear time
homomorphism counting in bounded degeneracy graphs~\cite{BePaSe20,BePaSe21}.
Gishboliner, Levanzov, and Shapira independently proved the same characterization using slightly different methods~\cite{GiLeSh20,BeGiLe+22}.

We give a short discussion of the Triangle Detection Conjecture.
Itai and Rodeh~\cite{ItRo78} gave the first non-trivial algorithm for the triangle 
detection and finding problem with $O(m^{3/2})$ runtime. 
The current best known algorithm runs in time 
$O(\min \{n^\omega, m^{{2\omega}/{(\omega+1)}}\})$~\cite{AlYuZw97}, where $\omega$ is the matrix multiplication
exponent. Even for $\omega = 2$, the bound is $m^{4/3}$ and widely believed to be a lower bound.
Many classic graph problems have fine-grained complexity hardness
based on \TRICONJ~\cite{AbWi14}.

Homomorphism or subgraph orbit counts have found significant use in network analysis and machine learning.
Przulj introduced the use of graphlet (or orbit count) degree distributions in bioinformatics~\cite{Pr07}.
The graphlet kernel of Shervashidze-Vishwanathan-Petri-Mehlhorn-Borgwardt uses vertex orbits counts
to get embeddings of vertices in a network~\cite{ShViPe+09}.
Four vertex subgraph and large cycle and clique orbit counts have been used for 
discovering special kinds of vertices and edges~\cite{UgBaKl13,RoKaKl+17,YiBiLe18,YiBiLe19}.
Orbits counts have been used to design faster algorithms for finding dense subgraphs in practice~\cite{BeHeLa+11,SaSePi+15,Ts15,BeGlLe16,TsPaMi17}.

\section{Preliminaries} \label{sec:preliminaries}

We use $G$ to denote the input graph and $H$ to denote the pattern graph, both \break $G=(V(G),E(G))$ and $H=(V(H),E(H))$ are simple, undirected and connected graphs. We denote $|V(G)|$ and $|E(G)|$ by $n$ and $m$ respectively and $|V(H)|$ by $k$.

A pattern graph $H$ is divided into orbits, we use the definition from Bondy and Murty (Chapter $1$, Section $2$ \cite{BoMu08}):

\begin{definition} \label{def:orbits}
	Fix a graph $H = (V (H), E(H))$. An automorphism
	is a bijection $\sigma : V (H) \to V (H)$ such that $(u,v) \in E(H)$ iff $(\sigma(u), \sigma(v)) \in E(H)$.
	The group of automorphisms of $H$ is denoted $\Aut(H)$.
	
	Define an equivalence relation on $V(H)$ as follows. We say
	that $u \sim v\ (u,v \in V (H))$ iff there exists an automorphism that maps
	$u$ to $v$. The equivalence classes of the relation are called orbits.
\end{definition}

We refer to the set of orbits in $H$ as $\Psi(H)$ and to individual orbits in $\Psi(H)$ as $\psi$. Note that every vertex $h \in V(H)$ belongs to exactly one orbit. 
We can represent an orbit by a canonical (say lexicographically least) vertex in the orbit. Somewhat abusing notation, we can think of the set
of orbits as a subset of vertices of $H$, where each vertex plays a ``distinct role'' in $H$.
\Fig{graphs_loipl} has examples of different graphs with their separate orbits.

We now define homomorphisms.

\begin{definition} \label{def:hom}
	An \emph{$H$-homomorphism} from $H$ to $G$ is a mapping $\phi: V(H) \to V(G)$ such that for all $(u,v) \in E(H)$, $(\phi(u),\phi(v))\in E(G)$. 
	We refer to the set of homomorphisms from $H$ to $G$ as $\Phi(H,G)$.
\end{definition}

We now define a series of counts.
\smallskip
\begin{itemize}
	\item $\Hom{G}{H}$: This is the count of $H$-homomorphisms in $G$.
	So $\Hom{G}{H} = |\Phi(H,G)|$.
	\item $\OHC_{H,\psi}(v)$: For a vertex $v \in V(G)$, $\OHC_{H,\psi}(v)$ is the number
	of \linebreak $H$-homomorphisms that map any vertex in the orbit $\psi$ to $v$. Formally, $\OHC_{H,\psi}(v) = |\{\phi \in \Phi(H,G) \colon \exists u\in \psi,\, \phi(u)=v\}|$.
	\item $\OHC_{H,\psi}(G), \OHC_H(G)$: We use $\OHC_{H,\psi}(G)$ to denote
	the list/vector of counts $\{\OHC_{H,\psi}(v)\}$ over all $v \in V(G)$.
	Similarly, $\OHC_H(G)$ denotes the sequence of lists of counts $\OHC_{H,\psi}(G)$ over all orbits $\psi$.
\end{itemize}
\medskip

Our aim is to compute $\OHC_H(G)$, which are a set of homomorphism counts.
We use existing algorithmic machinery to compute homomorphism counts per vertex of $H$, so part of our analysis
will consist of figuring out how to go between these counts. As we will see, this is where the \LIPCO{}
parameter makes an appearance.

\subparagraph*{Acyclic orientations} These are a key algorithmic tool in efficient algorithms for bounded degeneracy graphs.
An acyclic orientation of an undirected graph $G$ is a digraph obtained by directing the edges of $G$ 
such that the digraph is a DAG.
We will encapsulate the application of the degeneracy in the following lemma, which holds
from a classic result of Matula and Beck \cite{MaBe83}.

\begin{lemma} \label{lem:degen-ordering} Suppose $G$ has degeneracy $\degen$. Then, in $O(m+n)$ time,
	one can compute an acyclic orientation $\dirG$ of $G$ with the following property.
	The maximum outdegree of $\dirG$ is precisely $\degen$. ($\dirG$ is also called
	a \emph{degeneracy orientation}.)
\end{lemma}

The set of all acyclic orientations of $H$ is denoted $\Sigma(H)$. Our algorithm will enumerate over all such orientations.

Note that all definitions of homomorphisms carry over to DAGs.

\subsection{DAG-tree decompositions} \label{sec:dagtree}

A central part of our result is applying intermediate lemmas from an important algorithm
of Bressan for homomorphism counting~\cite{Br19}. This subsection gives a technical
overview of Bressan's techique of DAG-tree decompositions and related lemmas. Our aim is to state
the key lemmas from previous work that can be used as a blackbox.

The setting is as follows. We have an acyclic orientation $\dirG$ and a DAG pattern
$P$ (think of $P$ as a member of $\Sigma(H)$; $P$ is an acyclic orientation of $H$). 
Bressan's algorithm gives a dynamic programming approach to counting $\Phi(P,\dirG)$.

We introduce some notation. We use the standard notion of reachability in digraphs: vertex $v$
is reachable from $u$ if there is a directed path from $u$ to $v$.
\begin{itemize}
	\item $S$: The set of sources in the DAG $P$.
	\item $\Reachable_{P}(s)$: For source $s \in S$, $\Reachable_{P}(s)$ is the set of vertices in $P$ reachable from $s$.
	\item $\Reachable_{P}(B)$: Let $B \subseteq S$. $\Reachable_{P}(B) = \bigcup_{s \in B} \Reachable_P(s)$.
	\item $P[B]$: This is the subgraph of $P$ induced by $\Reachable_P(B)$.
\end{itemize}

\medskip

\begin{definition} \label{def:dag-tree}
	(DAG-tree decomposition \cite{Br19}). Let $P$ be a DAG with source vertices $S$. A
	DAG-tree decomposition of P is a tree $T = (\mathcal{B}, \mathcal{E})$ with the following three properties:
	\begin{enumerate}
		\item Each node $B \in \mathcal{B}$ (referred to as a ``bag'' of sources) is a subset of the source vertices $S$: $B\subseteq S$.
		\item The union of the nodes in $T$ is the entire set $S$: $\bigcup_{B\in\mathcal{B}}B = S$.
		\item For all $B, B_1, B_2 \in \mathcal{B}$, if $B$ lies on the unique path between the nodes $B_1$ and $B_2$ in $T$, then $\Reachable(B_1) \cap \Reachable(B_2) \subseteq \Reachable(B)$.
	\end{enumerate}
\end{definition}

\begin{definition} \label{def:treewidth} Let $P$ be a DAG. For any DAG-tree decomposition $T$ to $P$, the
	DAG-treewidth $\tau(T)$ is defined as $\max_{B \in \mathcal{B}} |B|$.
	The DAG-treewidth of $P$, denoted $\tau(P)$, is the minimum value of $\tau(T)$ over all DAG-tree decompositions
	$T$ of $P$.
\end{definition}

\subparagraph*{Two important lemmas} We state two critical results from previous work. Both of these are highly non-trivial
and technical to prove. We will use them in a black-box manner. 
The first lemma, by Bera-Pashanasangi-Seshadhri, connects the Largest Induced Cycle Length (LICL)
to DAG-treewidth~\cite{BePaSe21}.

\begin{lemma} \label{lem:bera}
	{(Theorem 4.1 in \cite{BePaSe21})} For a simple graph $H$: $LICL(H)\leq 5$ iff $\forall P \in \Sigma(H), \tau(P)=1$.
\end{lemma}

The second lemma is an intermediate property of Bressan's subgraph counting algorithm~\cite{Br21}. We begin by defining
homomorphism extensions. Think of some directed pattern $P$ that we are trying to count. Fix a (rooted) DAG-tree decomposition $T$.
Let $P'$ be a subgraph of $P$, $P''$ be a subgraph of $P'$. A $P'$-homomorphism $\phi'$
\emph{extends} a $P''$-homomorphism $\phi''$ if $\forall v \in V(P'')$, $\phi'(v) = \phi''(v)$. Basically,
$\phi'$ agrees with $\phi''$ wherever the latter is defined.

\smallskip

\begin{itemize}
	\item $\ext(P',G;\phi)$: Let $\phi$ be a homomorphism from a subgraph of $P'$ to $G$.
	Then $\ext(P',G;\phi)$ is the number of $P'$-homomorphisms extending $\phi$.
	\item $P[\down(B)]$: Let $B$ be a node in the DAG-tree decomposition $T$ of $P$. The set $\down(B)$
	is the union of bags that are descendants of $B$ in $T$. Furthermore $P[\down(B)]$ is the pattern
	induced by $\Reachable(\down(B))$.
\end{itemize}

A technical lemma in Bressan's result shows that extension counts can be obtained efficiently.
We will refer to the procedure in this lemma as ``Bressan's algorithm''.

\begin{lemma} \label{lem:bressan}
	{(Lemma 5 in \cite{Br21})}: Let $G^{\to}$ be a digraph with outdegree at most $d$ and $P$ be a DAG with $k$ vertices.
	Let $T = (\mathcal{B}, \mathcal{E})$ be a DAG-tree decomposition for $P$, and $B$ any element of $\mathcal{B}$.
	There is a procedure, that in time $O(|\mathcal{B}| poly (k)d^{k-\tau(T)} n^{\tau(T)} \log n)$, returns
	a dictionary storing the following values: for every $\phi: P[B] \to G^{\to}$,
	it has $\ext(P[\down(B)],G;\phi)$.
\end{lemma}

Let us explain this lemma in words. For any bag $B$, which is a set of sources in $P$,
consider $P[B]$, which is the subgraph induced by $\Reachable_P(B)$. For every
$P[B]$-homomorphism $\phi$, we wish to count the number of extensions to $P[\down(B)]$
(the subgraph induced by vertices of $P$ reachable by any source in any descendant bag of $B$).

\section{Obtaining Vertex-Centric Counts} \label{sec:vhc}

We define vertex-centric
homomorphism counts, which allows us to ignore orbits and symmetries in $H$.
Quite simply, for vertices $h \in V(H)$ and $v \in V(G)$, we count the number of homomorphisms from $H$ to $G$ that map $h$ to $v$.

\begin{definition} \label{def:vhc}
	{Vertex-centric Counts}: For each vertex $h \in V(H)$ and vertex $v \in V(G)$, let $\VHC_{H,h}(v)$
	be the number of $H$-homomorphisms that map $h$ to $v$.
	
	Let $\VHC_H(G)$ denote the list of $\VHC_{H,h}(v)$ over all $h \in V(H)$ and $v \in V(G)$.
\end{definition}

We can show that the vertex-centric counts can be obtained in near-linear time when $LICL(H) \leq 5$:

\begin{theorem} \label{thm:vhc_upper_2} There is an algorithm that takes as input 
	a bounded degeneracy graph $G$ and a pattern $H$ with $LICL(H) \leq 5$, and has the following properties.
	It outputs $\VHC_H(G)$ and runs in $O(n\log{n})$ time.
\end{theorem}

Before proving this theorem we need to introduce two more lemmas. First, we invoke the following lemma from \cite{Br21}:

\begin{lemma}
	\label{lem:bressan_construction}
	{(Lemma 4 in \cite{Br21})}: Given any $B \subseteq S$, the set of homomorphisms $\Phi(P[b],\dirG)$ has size $O(d^{k-|B|}n^{|B|})$ and can be enumerated in time $O(k^2 d^{k-|B|}n^{|B|})$.
\end{lemma}

Second, we show how to use the output of Bressan's algorithm to obtain the Vertex-centric counts:

\begin{lemma} \label{lem:bressan_adapted}
	Let $P$ be a directed pattern on $k$ vertices, $T=(\mathcal{B},\mathcal{E})$ be a DAG-tree decomposition of $P$ with $\tau(P)=1$ (All nodes/bags in $T$ are singletons),
	and $\dirG$ be a directed graph with $n$ vertices and max degree $d$. Let $b$ be the root of $T$ and $h$ be any vertex in $P[b]$.
	We can compute $\VHC_{P,h}(v)$ in time $O(poly(k) d^{k-1}n\log{n})$.
\end{lemma}
\begin{proof}
	The algorithm of \Lem{bressan} will return a data structure/dictionary that gives the following values. For each $\phi: P[b] \to \dirG$, it provides $\ext(P[\down(b)],G;\phi)$. 
	Note that $b$ is the root of $T$. By the properties of a DAG-tree decomposition, $\down(b)$ contains all vertices of $P$ and $P[\down(b)] = P$. 
	Hence, the dictionary gives the values $\ext(P,\dirG;\phi)$, that is, the number of homomorphisms $\phi': P \to \dirG$ that extend $\phi$.
	
	Let $h$ be a vertex in $P[b]$. We can partition the set of homomorphisms from \break $P[b]$ to $\dirG$, $\Phi(P[b],\dirG)$, into sets $\Phi_{b,h,v}$ defined as follows.
	For each $v \in V(G)$,\break $\Phi_{b,h,v} \eqdef \{\phi \in \Phi(P[b],\dirG) : \phi(h) = v\}$.
	
	By Lemma \ref{lem:bressan_construction} we can list all the homomorphisms $\Phi(P[b],\dirG)$ in $O(k^2 d^{k-1}n)$ time, by the same lemma we know that $\Phi(P[b],\dirG)$ will have size at most $O(d^{k-1}n)$, hence we can iterate over the list of homomorphisms and check the value of $\phi(h)$. We can then express $\VHC_{P}(\dirG)$ as follows:
	\begin{align*}
		\VHC_{P,h}(v) 
		&= |\{\phi' \in \Phi(P,\dirG): \phi'(h)=v\}|
		\\ 
		&= \sum_{\phi \in \Phi(P[b], \dirG) : \phi(h)=v} ext(P,\dirG;\phi)
		\\
		&= \sum_{\phi \in \Phi_{b,h,v} : \phi(h)=v} ext(P,\dirG;\phi)
	\end{align*}
	
	We can compute all of these values by enumerating all the elements in $\phi \in \Phi_{b,h,v}$ (over all $v$), and making
	a dictionary access to get $ext(P,\dirG;\phi)$. 	
	The total running time is $O(k^2d^{k-1}n \log{n})$, where $\log{n}$ is extra overhead of accessing the dictionary.
	
	By \Lem{bressan}, the dictionary construction takes $O(|\mathcal{B}| poly (k)$ $d^{k-\tau(T)} n^{\tau(T)} \log n)$ time.
	Since $\tau(T)=1$ and $|\mathcal{B}| = O(k)$, we can express the total complexity as $O(poly(k)d^{k-1}n\log{n})$.
\end{proof}

We can now complete the proof of Theorem \ref{thm:vhc_upper_2}:

\begin{algorithm}
	\caption{ \textbf{Vertex-centric Counts:} \\ $\ComputeVHC(H,G)$}\label{alg:counts}
	\begin{algorithmic}[1]
		\State Initialize all $\VHC_{H,h}(v)$ counts to $0$ ($\forall h \in V(H), v \in V(G)$).
		\State Compute the degeneracy orientation $\dirG$ of $G$.
		\For{$P \in \Sigma(H)$}
		\State Initialize all $\VHC_{P,h}(v)$ counts to $0$.
		\State Compute width $1$ DAG-tree decomposition $T=(\mathcal{B},\mathcal{E})$.
		\For{$b \in \mathcal{B}$}
		\State Root $T$ at $b$.
		\State Use \Lem{bressan} to compute the dictionary.
		\For{$h \in P[b]$}
		\For{$v \in V(G)$}
		\If{$\VHC_{P,h}(v)$ not updated}
		\For{$\phi \in \Phi_{b,h,v}$}
		\State Compute $ext(P,\dirG;\phi)$ (dict. lookup). 
		\State Add $ext(P,\dirG;\phi)$ to $\VHC_{P,h}(v)$.
		\EndFor
		\State Mark $\VHC_{P,h}(v)$ as updated
		\EndIf
		\EndFor
		\EndFor
		\EndFor
		\State For each $h \in V(H)$ and $v \in V(G)$, add $\VHC_{P,h}(v)$ to $\VHC_{H,h}(v)$.
		\EndFor
	\end{algorithmic}
\end{algorithm}

\begin{proof}[Proof of Theorem \ref{thm:vhc_upper_2}]
	
	The algorithm is explicitly described in Algorithm \ref{alg:counts}.
	
	The first step of our algorithm is to construct the degeneracy orientation $\dirG$ of $G$. By \Lem{degen-ordering}, it
	can be computed in $O(m+n)$ time. Since $G$ has bounded degeneracy, $\dirG$ has bounded outdegree. When orienting $G$ as $\dirG$, 
	each homomorphism from $H$ to $G$ becomes a homomorphism of exactly one of the directed patterns $P\in\Sigma(H)$ to $\dirG$. We can hence compute $\VHC_H(G)$ as the sum of $\VHC_{P}(\dirG)$ for every acyclic orientation of $H$. This is given by the following equation:
	
	\begin{align} \label{eq:combining_orientations}
		\VHC_{H}(G) = \sum_{P \in \Sigma(H)} \VHC_{P}(G^{\to})
	\end{align}
	
	Because $LICL(H) \leq 5$, \Lem{bera} implies that for all $P \in \Sigma(H)$, $\tau(H) = 1$. There exists a DAG-tree decomposition $T=(\mathcal{B},\mathcal{E})$ of $P$ with $\tau(T)=1$. 
	We use the output of Bressan's algorithm to obtain the Vertex-centric counts.

	The DAG-tree decomposition $T$ can be arbitrarily rooted at any node $b$. Moreover, for each $h \in V(P)$, 
	there must exist some source $b$ such that $h \in P[b]$ (meaning, $h$ is reachable from $b$). So, by rooting $T$
	at all possible nodes (singleton bags), we can ensure that $h$ is in $P[b]$. We can apply \Lem{bressan_adapted}
	to get all counts $\VHC_{P,h}(v)$.
	
	We complete the proof by bounding the running time and asserting correctness.
	
	From Lemma \ref{lem:degen-ordering}, we can compute $\dirG$ in $O(m+n)$.
	Since $G$ has bounded degeneracy, $m = O(n)$ and the outdegree $d$ is bounded.
	The number of acyclic orientations of $H$, $|\Sigma(H)|$ is bounded by $O(k!)$. In each iteration, by \Lem{bressan_adapted}, we will take $O(poly(k) d^{k-1}n\log{n})$. 
	For constant $k$ and constant $d$, the running time is $O(n\log n)$.
	
	Now we prove the correctness of the algorithm. Consider each $P \in \Sigma(H)$.
	Let $T=(\mathcal{B},\mathcal{E})$ be the DAG-tree decomposition of $P$. For each $b \in \mathcal{B}$, we compute $\VHC_{P,h}(\dirG)$ for all the vertices in $h \in P[b]$. 
	By looping over each singleton bag $b$, we update counts for all vertices in $P$. Hence, we are computing $\VHC_{P}(\dirG)$. Finally, we sum over all $P\in \Sigma(H)$, which by Equation \ref{eq:combining_orientations}, gives us $\VHC_{H}(G)$.
	
\end{proof}

\section{From Vertex-Centric to Orbit Counts} \label{sec:orbit-dag}

We now show how to go from vertex-centric to orbit counts, using inclusion-exclusion. Much of our insights are given
by the following definitions.

\begin{definition} \label{def:IS}
	$\IS(\psi)$: Given a pattern graph $H$, for every orbit $\psi \in \Psi(H)$ we define $\IS(\psi)$ as the collection of all non empty subsets $S \subseteq \psi$, such that $S$ forms an independent set (i.e. there is no edge in $E(H)$ connecting any two vertices in $S$). 
	
	Formally, $\IS(\psi) = \{ S \subseteq \psi, S \neq \emptyset: \forall\ h,h' \in S, (h,h')\notin E(H)\}$.
\end{definition}
\begin{definition} \label{def:HS}
	$H_S$: For each set $S \in \IS(\psi)$ we define $\HS$ as the graph resulting from merging all the vertices in $S$ into a single new vertex $h_S$, removing any duplicate edge.
\end{definition}

We state two more tools in our analysis. The first lemma relates the counts obtained
in the previous section ($\VHC_{\HS}(G)$) to the desired output ($\OHC_H(G)$).

\begin{lemma} \label{lemma:vhc_to_ohc}
	{(Inclusion-exclusion formula:)}
	\begin{equation*}
		\OHC_{H,\psi}(v) = \sum_{S \in \IS(\psi)} (-1)^{|S|+1} \VHC_{\HS,h_S}(v)
	\end{equation*}
\end{lemma}

In order to prove this lemma, we need to define the Signature of a homomorphisms. Let $\phi$ be a homomorphism from $H$ to $G$, we define $\Signature(\phi, \psi, v)$ to be the subset of vertices from the orbit $\psi$ that are mapped to $v$ in $\phi$. Formally $\Signature(\phi, \psi, h) = \{h \in \psi : \phi(h) = v\}$. 

We prove a series of claims regarding the signature.

\begin{claim} \label{claim:signature_IS}
	The Signature of $\phi$ from $\psi$ to $v$, $\Signature(\phi, \psi, v)$, must form an independent set of vertices in $V(H)$, that is, there are no edges in $E(H)$ connecting two vertices in $\Signature(\phi, \psi, v)$.
\end{claim}
\begin{proof}
	We prove by contradiction. Assume that  $S= \Signature(\phi, \psi, v)$ is not an Independent Set of vertices of $V(H)$, that means that we have a pair of vertices $h,h' \in S$ such that there is an edge connecting them. But from the definition of signature we have that $\phi(h) = \phi(h') = v$, however this is not a valid homomorphism from $H$ to $G$ as it is not preserving the $(h,h')$ edge.
\end{proof}

The next claim allows us to relate the Signature with the Homomorphism Orbit Counts. 
\begin{claim} \label{claim:OHC_to_Signature}
	\[
	\OHC_{H,\psi}(v) = \sum_{S \in \IS(\psi)} |\{\phi \in \Phi(H,G) : S=\Signature(\phi,\psi,v)\}|
	\]
\end{claim}
\begin{proof}
	From the definition of Homomorphism Orbit Counts we have that $\OHC_{H,\psi}(v) = |\{\phi \in \Phi(H,G) \ : \exists h\in \psi,\ \ \phi(h)=v\}|$. Hence, suffices to show that $|\{\phi \in \Phi(H,G) \ : \exists h\in \psi,\ \ \phi(h)=v\}| = \sum_{S \in \IS(\psi)} |\{\phi \in \Phi(H,G) : S=\Signature(\phi,\psi,v)\}|$.
	
	Let $\phi \in \Phi(H,G)$ be a homomorphism from $H$ to $G$ such that $\exists h\in \psi, \phi(h)=v$. Let $S=\Signature(\phi,\psi,v)$, we know that $S \neq \emptyset$ as $h$ is mapped to $v$ and from Claim \ref{claim:signature_IS} we know that it forms an independent set on the vertices of $H$. Hence $S \in \IS(\psi)$. 
	
	To prove the other direction of the equality, suffices to note that if a homomorphism $\phi$ contributes to the right side of the equation, then its signature $S$ belongs to $\IS(\psi)$, hence there is at least one vertex $h \in V(H)$ that is mapped to $v$, and thus $\phi$ contributes to the left side of the equation.
\end{proof}

Now, we will relate the Signature with the Vertex-centric Counts:

\begin{claim} \label{claim:signature_group} 
	For each orbit $\psi$ in $H$ and each vertex $v$ in $V(G)$ we have that $\forall\ S\in \IS(\psi)$:
	\begin{equation*}
		|\phi \in \Phi(H,G) : \forall h\in S, \phi(h)=v| = \sum_{\substack{S' \supseteq S \\ S'\in \IS(\psi)}} |\phi : \Signature(\phi,\psi,v)=S'|
	\end{equation*}
\end{claim}
\begin{proof}
	If $\phi$ is mapping all the vertices in $S$ to $v$, then the Signature of $\phi$ from $\psi$ to $v$ must be a superset of $S$, $\Signature(\phi,\psi,v) \supseteq S$. Hence summing over such sets will reach the equality. Note that we can add the restriction of $S'$ belonging to $\IS(\psi)$ as it is implied from Claim \ref{claim:signature_IS}.
\end{proof}

Let $\Phi'=\Phi(\HS, G)$ be the set of homomorphism from $\HS$ to $G$. When $S$ forms an independent set there is an equivalence between the homomorphisms in $\Phi'$ that map $h_S$ to $v$ and the set of homomorphisms in $\Phi(H,G)$ that map all the vertices of $S$ to $v$. In fact we can prove the following claim:

\begin{claim} \label{claim: vhc_H/S} If $S$ is not empty and form an independent set:
	\begin{equation*}
		|\phi \in \Phi(H,G) :\ \forall\ h\in S\ \phi(h)=v| = \VHC_{\HS,h_S}(v)
	\end{equation*}
\end{claim}
\begin{proof}
	From the definition of Vertex-centric Counts we have that $\VHC_{\HS,h_S}(v)= |\phi' \in \Phi(\HS,G): \phi(h_S)=v|$. Hence it suffices to show that:
	\begin{equation*}
		\left|\phi \in \Phi(H,G) : \forall h\in S\ \phi(h)=v\right| = \left|\phi' \in \Phi(\HS,G): \phi(h_S)=v\right|               
	\end{equation*}

	We do so by proving that there is a bijection between both sets, that is, a one to one correspondence between them. Let $\Phi_S = \{\phi' \in \Phi(\HS,G): \phi(h_S)=v\}$ and \break $\Phi'_S = \{\phi \in \Phi(H,G) :\forall h\in S,\phi(h)=v\}$. We show an invertible function $f: \Phi_S \to \Phi'_S$:
	\begin{itemize}
		\item Given a homomorphism $\phi \in \Phi_S$ we  obtain $\phi'=f(\phi) \in \Phi'_S$ by setting $\phi'(h)=\phi(h)\ \forall\ h\in H\setminus S$ and $\phi'(h_S)=v$. This is a valid homomorphism as we are mapping all the vertices in $\HS$ to $G$ and we are preserving the edges.
		
		\item Given a homomorphism $\phi' \in \Phi'_S$ we obtain $\phi=f'(\phi') \in \Phi_S$ by setting $\phi(h)=\phi'(h)\ \forall\ h\in H\setminus S$ and $\phi(h)=v\ \forall \ h \in S$. Again this is a valid homomorphism as we are mapping all the vertices in $H$ to $G$ and we are still preserving the edges. 
	\end{itemize}
	Additionally, we have that for all $\phi \in \Phi_S$, $\phi = f'(f(\phi))$, which completes the proof.
\end{proof}

We will show one last claim that will be important when deriving the inclusion-exclusion formula:
\begin{claim} \label{claim:sums}
	Given a graph $H$, for every orbit $\psi \in \Psi(H)$, any subset $S'\in \IS(\psi)$ satisfies:
	\begin{align*}
		\sum_{\substack{S \subseteq S' \\ S\neq \emptyset}}  (-1)^{|S|+1} = 1
	\end{align*}
\end{claim}
\begin{proof}
	\begin{align*}
		&\sum_{\substack{S \subseteq S' \\ S\neq \emptyset}}  (-1)^{|S|+1} 
		= \sum_{i=1}^{|S'|} \binom{|S'|}{i}(-1)^{i+1} 
		\\
		&= \sum_{i=1}^{|S'|} \left(\binom{|S'|-1}{i-1} + \binom{|S'|-1}{i}\right)(-1)^{i+1} 
		= \binom{|S'-1|}{0}(-1)^2 
		= 1
	\end{align*}
\end{proof}

We now have all the tools required to prove Lemma \ref{lemma:vhc_to_ohc}:

\begin{proof}[Proof of Lemma \ref{lemma:vhc_to_ohc}]
	\begin{align*}
		&\sum_{S \in \IS(\psi)} (-1)^{|S|+1} \VHC_{\HS,h_S}(v)
		\\
		&=\sum_{S \in \IS(\psi)} (-1)^{|S|+1} |\phi \in \Phi(H,G) :\ \forall\ h\in S\ \phi(u)=v| && \text{(Claim \ref{claim: vhc_H/S})}
		\\
		& = \sum_{S \in \IS(\psi)} (-1)^{|S|+1} \sum_{\substack{S' \supseteq S \\ S'\in \IS(\psi)}} |\phi : \Signature(\phi,\psi,v)=S'| && \text{(Claim \ref{claim:signature_group})}
		\\
		& = \sum_{S \in \IS(\psi)} \sum_{\substack{S' \supseteq S \\ S'\in \IS(\psi)}}  (-1)^{|S|+1} |\phi : \Signature(\phi,\psi,v)=S'| && \text{(Factor in)}
		\\
		& = \sum_{S' \in \IS(\psi)} \sum_{\substack{S \subseteq S' \\ S\neq \emptyset}}  (-1)^{|S|+1} |\phi : \Signature(\phi,\psi,v)=S'| && \text{(Reorder)}
		\\
		&  = \sum_{S' \in \IS(\psi)} |\phi : \Signature(\phi,\psi,v)=S'| \sum_{\substack{S \subseteq S' \\ S\neq \emptyset}}  (-1)^{|S|+1} && \text{(Factor out)}
		\\
		& =  \sum_{S' \in \IS(\psi)} |\phi : \Signature(\phi,\psi,v)=S'| && \text{(Claim \ref{claim:sums})}
		\\
		& = \OHC_{H,\psi}(v) && \text{(Claim \ref*{claim:OHC_to_Signature})}
	\end{align*}
\end{proof}

The next lemma relates the Longest Induced Path Connecting Orbits (\LIPCO{}) defined in \Def{lipo} with the $LICL$ of all the graphs $\HS$,
for all $S \in \IS(\psi)$ and all orbits $\psi$ of $H$.

\begin{lemma} \label{lemma:lipco_to_licl}
	For every graph $H$, $\LIPCO(H) \leq 5$ iff $\ \forall \psi \in \Psi(H), \forall S \in \IS(\psi)$, \linebreak $LICL(\HS)\leq 5$.
\end{lemma}
\begin{proof}
	First, we show that if $\LIPCO(H) > 5$ then $\exists \psi \in \Psi(H), \exists S \in \IS(\psi),  LICL(\HS) > 5$. Consider the longest induced path in $H$ with endpoints in the same orbit $\psi \in \Psi(H)$, let $h,h'$ be the two endpoints of the path. We have two cases:
	\begin{itemize}
		\item $h=h'$: In this case the induced path is actually just an induced cycle of length $6$ or more in $H$ including the vertex $h$. For any $\psi$ and for any $S \subseteq \psi$ with $|S|=1$ we have that $\HS=H$, and hence $LICL(\HS)>5$.
		
		\item $h\neq h'$: In the other case we have that $h,h'$ are distinct vertices. Consider the set $S=\{h,h'\}$, we have that $S \in \IS(\psi)$ as both $h,h' \in \psi$ and there is no edge connecting them (otherwise we would have a longer induced cycle). We form $\HS$ by combining $h$ and $h'$ into a single vertex, the induced path that we had in $H$ becomes then an induced cycle of length at least $6$, which implies $LICL(\HS)>5$.
	\end{itemize} 
	
	Now, we prove that if $\exists \psi \in \Psi(H), \exists S \in \IS(\psi), LICL(\HS) > 5$ then $\LIPCO(H) > 5$. Let $S$ be the set such that $LICL(\HS)>5$. Again, we have two cases:
	\begin{itemize}
		\item $|S|= 1$: We have that $\HS=H$ and hence $LICL(H)>5$, any vertex in that induced cycle induces a path of the same length with such vertex in both ends, which implies $\LIPCO(H) > 5$. 
		
		\item $|S|> 1$: Let $h_S$ be the vertex in $\HS$ obtained by merging the vertices of $S$ in $H$. Consider the longest induced cycle in $\HS$, if that cycle does not contain $h_S$ then that same cycle exists in $H$ and $LICL(H)>5$, which implies $\LIPCO(H) > 5$. Otherwise, we can obtain $H$ by splitting $h_S$ back into separate vertices, there will be two distinct vertices $h,h' \in S$ that are in the two ends of an induced path of the same length in $H$, thus $\LIPCO(H) > 5$.
	\end{itemize}
\end{proof}

\section{Wrapping it up} \label{sec:wrapup}

In this section we complete the proof of the main theorem for the upper bound. We also give Algorithm \ref{alg:exclusion_inclusion}, which summarizes the entire process.

\begin{theorem} \label{thm:ohc_upper_2}
	There is an algorithm that, given a bounded degeneracy graph $G$ and pattern $H$ with $LIPCO(H) \leq 5$, computes $\OHC_{H}(G)$ in time $O(n\log{n})$.
\end{theorem}
\begin{proof}
	Because we have that $\LIPCO(H) \leq 5$, using Lemma \ref{lemma:lipco_to_licl} we get that $\forall \psi \in \Psi(H), \forall S \in \IS(\psi), LICL(\HS) \leq 5$. This means, using Theorem \ref{thm:vhc_upper_2}, that $\forall \psi \in \Psi(H), \forall S \in \IS(\psi)$ we can compute $\VHC_{\HS}(G)$ in time $f(k)O(n \log{n})$.
	
	Using Lemma \ref{lemma:vhc_to_ohc} we can compute $\OHC_{H}(G)$ from the individual counts of \linebreak $\VHC_{\HS}(G)$ (as shown in Algorithm \ref{alg:exclusion_inclusion}), we have at most $2^k$ sets $S$, hence the total time complexity necessary to compute $\OHC_{H}(G)$ is $O(n \log{n})$.
\end{proof}

\begin{algorithm}
	\caption{\textbf{Homomorphism Orbit Counts }$\OHC_{H}(G)$}\label{alg:exclusion_inclusion}
	\begin{algorithmic}[1]
		\For{each $\psi \in \Psi(H)$} 
		\For{$S \in \IS(\psi)$}
		\State Compute $\VHC_{\HS,h_S}(G)$
		\EndFor
		\State $\OHC_{H,\psi}(G) = \sum_{S \in IS(\psi)} (-1)^{|S|+1} \VHC_{\HS,h_S}(v)$
		\EndFor
		\State Return $\OHC_{H}(G)$
	\end{algorithmic}
\end{algorithm}

\section{Lower Bound for computing Homomorphism Orbit Counts}\label{sec:lower_bound}

In this section we prove the lower bound of Theorem \ref{thm:main}. It will be given by the following theorem:

\begin{theorem} \label{thm:ohc_lower}
	Let $H$ be a pattern graph on $k$ vertices with  $\LIPCO(H) > 5$. Assuming the Triangle Detection Conjecture, there exists an absolute constant $\gamma>0$ such that for any function $f : \mathbb{N} \times \mathbb{N} \rightarrow \mathbb{N}$, there is no (expected) $f(\kappa,k)O(m^{1+\gamma})$ algorithm for the $\OHC$ problem, where $m$ and $\kappa$ are the number of edges and the degeneracy of the input graph, respectively.
\end{theorem}

To prove this Theorem we will show how to express the Homomorphism Orbit Counts for some orbit $\psi$ as a linear combination of Homomorphism counts of non-isomorphic graphs $\HS$ for all $S$ in $\IS(\psi)$. Because $\LIPCO(H) > 5$ we will have that the $LICL$ of at least one of these graphs is also greater than $5$. We will then show that the hardness of computing Orbit counts in the original graph is the same than the hardness of computing the Homomorphisms counts. Finally we use a previous hardness result from \cite{BePaSe21} to complete the proof.

First, we introduce the following definition:

\begin{definition} \label{def:agg}
	Given a pattern graph $H$ and an input graph $G$, for the orbit $\psi$ of $H$, we define $\Agg(H,G,\psi)$ as the sum over every vertex $v \in V(G)$ of homomorphisms that are mapping some vertex in $\psi$ to $v$, that is:
	\[
	\Agg(H,G,\psi)=\sum_{v\in V(G)} \OHC_{H,\psi}(v)
	\]
\end{definition}

Note that if we can compute $\OHC_{H,\psi}(v)$ for every vertex $v$ in $G$ then we can also compute $\Agg(H,G,\psi)$ in additional linear time. Now, we state the following lemma:

\begin{lemma} \label{lemma:aggregating}
	For every pattern graph $H$ and every orbit $\psi \in \Psi(H)$, there is some number $l = l(H)$ such that the following holds. For every graph $G$ there are some graphs $G_1, ... , G_l$, computable in time $O(|V (G)| + |E(G)|)$, such that $|V(G_i)| = O(|V|)$ and $|E(G_i)| = O(|E|)$ for all $i = 1, ..., l$, and such that knowing $\Agg(H, G_1,\psi), ... , \Agg(H, G_l,\psi)$ allows one to compute $\Hom{G}{\HS}$ for all $S\in \IS(\psi)$, in time $O(1)$. Furthermore, if $G$ is $O(1)$-degenerate, then so are $G_1,...,G_l$.
\end{lemma}

First, we can  relate the Homomorphism Vertex Counts of a vertex $h \in V(H)$ to Homomorphism Counts from $H$ to $G$, as given in the following claim:

\begin{claim} \label{claim:vhc_sum} For all $h\in V(H)$:
	\begin{equation*}
		\sum_{v \in V(G)} \VHC_{H,h}(v) = \Hom{G}{H}
	\end{equation*}
\end{claim}
\begin{proof}
	\begin{align*}
		&\sum_{v \in V(G)} \VHC_{H,h}(v)
		\\
		&=\sum_{v \in V(G)} |\{  \phi \in \Phi(H,G): \phi(h) = v \}| && \text{(Def. of $\VHC$)}
		\\
		&=\left|\{  \phi \in \Phi(H,G): \phi(h) \in V(G) \}\right| && \text{(Sum over whole set)}
		\\
		&=\left|\Phi(H,G)\right| && \text{($\forall \phi: \phi(u)\in V(G)$)}
		\\
		&= \Hom{G}{H} && \text{(Def. of Hom)}
	\end{align*}
\end{proof}

We now state the following Lemma from \cite{BeGiLe+22}:

\begin{lemma}  \label{lemma:bera_4_2}
	(Lemma 4.2 from \cite{BeGiLe+22}): Let $H_1, ... , H_l$ be pairwise non-isomorphic graphs and let $c_1, ... , c_l$ be non-zero constants. For every graph $G$ there are graphs $G_1, ... , G_l$, computable in time $O(|V (G)|+|E(G)|)$, such that $|V (G_i)| =
	O(|V (G)|)$ and $|E(G_i)| = O(|E(G)|)$ for every $i = 1, ... , l$, and such that knowing $b_j := c_1 \cdot \Hom{G_j}{H_1} +...+ c_l \cdot \Hom{G_j}{H_l}$ for every $j = 1, ... , l$ allows one to compute $\Hom{G}{H_1}, ... , \Hom{G}{H_l}$ in time $O(1)$. Furthermore, if $G$ is $O(1)$-degenerate, then so are $G_1, ... , G_l$.
\end{lemma}

We will apply the previous lemma in a similar way as it is used the proof of Lemma $4.1$ in \cite{BeGiLe+22}.

\begin{proof} [Proof of Lemma \ref{lemma:aggregating}]
	Let $H_1,...,H_l$ be an enumeration of all the graphs $\HS$ for all $S\in \IS(\psi)$, up to isomorphism. This means that $H_1,...,H_l$ are pairwise non-isomorphic and $\{H_1, ... , H_l\} = \{\HS : S \in \IS(\psi)\}$. 
	
	Let $f(i) =  (-1)^{|S|+1}|\{S \in \IS(\psi) : \HS \text{ is isomorphic to } H_i\}|$ be the number of sets $S \in \IS(\psi)$ such that $H_S$ is isomorphic to $H_i$, with the sign being $(-1)^{|S|+1}$. Note that all such sets have equal $|S|$ and that the value of $f(i)$ is always non-zero. We will use $h_i$ to denote the vertex of $H_i$ that correspond to the vertices $h_S$ of the graphs $\HS$ that are isomorphic to $H_i$. We can express $Agg(H,G,\psi)$ as follows:.
	
	\begin{align*}
		&Agg(H,G,\psi) = \sum_{v\in V(G)} \OHC_{H,\psi}(v) && \text{(Def. \ref{def:agg})}\\
		&= \sum_{v\in V(G)}\sum_{S \in \IS(\psi)} (-1)^{|S|+1} \VHC_{\HS, h_S}(v) && \text{(Lemma \ref{lemma:vhc_to_ohc})}\\
		&= \sum_{v\in V(G)}\sum_{i=1}^l f(i) \VHC_{H_i, h_i}(v)&& \text{(Def. of $f(i)$)}\\
		&= \sum_{i=1}^l f(i)\sum_{v\in V(G)} \VHC_{H_i, h_i}(v)&& \text{(Reorder)}\\
		&= \sum_{i=1}^l f(i) \Hom{G}{H_i} && \text{(Claim \ref{claim:vhc_sum})}
	\end{align*}
	Hence, we have that $\Agg(H,G,\psi)$ is a linear combination of homomorphism counts of $H_1,...,H_l$. We can then use Lemma \ref{lemma:bera_4_2} to complete the proof.
\end{proof}

Before we prove Theorem \ref{thm:ohc_lower}, we need to state the following theorem from \cite{BePaSe21}, which gives a hardness result on Homomorphism Counting:

\begin{theorem} [Theorem 5.1 from \cite{BePaSe21}]\label{thm:bera_lower}
	Let $H$ be a pattern graph on $k$ vertices with $LICL \geq 6$. Assuming the Triangle Detection Conjecture, there exists an absolute constant $\gamma$ such that for any function $f : \mathbb{N} \times \mathbb{N} \rightarrow \mathbb{N}$, there is no (expected) $f(\kappa,k)O(m^{1+\gamma})$ algorithm for the $\HomAl{H}$ problem, where $m$ and $\kappa$ are the number of edges and the degeneracy of the input graph, respectively.
\end{theorem}

We now have all the tools required to proof Theorem \ref{thm:ohc_lower}:

\begin{proof}[Proof of Theorem \ref{thm:ohc_lower}]
	We prove by contradiction. Given a graph $G$ and a pattern $H$ with $\LIPCO(H) > 5$, suppose there exists an algorithm that allows us to compute $\OHC_{H}(G)$ in time $f(\kappa,k)O(m)$, by Lemma \ref{lemma:aggregating} we have the existence of some graphs $G_1,...,G_l$. We can compute $\OHC_{H}(G_i)$ for all of these graphs in time $f(\kappa,k)O(m)$ and then aggregate the results into $\Agg(H,G_i,\psi)$ for all $G_i$ and all $\psi \in \Psi(H)$. Using Lemma \ref{lemma:aggregating}, that implies that we can compute $\Hom{G}{\HS}$ for all $S \in \IS(\psi)$ for all $\psi \in \Psi(H)$ in time $f(\kappa,k)O(m)$. 
	
	However, if $\LIPCO{}(H)> 5$ then, by Lemma \ref{lemma:lipco_to_licl}, we have that there exists a $S\subseteq \psi$ for some $\psi \in \Psi(H)$ such that $LICL(\HS)> 5$. From Theorem \ref{thm:bera_lower} we know that in that case there is no algorithm that computes $\Hom{G}{\HS}$ in time $f(\kappa,k)O(m^{1+\gamma})$ for some constant $\gamma>0$. This is a contradiction, and hence no algorithm can compute $\OHC_{H}(G)$ in $f(\kappa,k)O(m)$ time. 
\end{proof}

\bibliography{cleanbib}

\end{document}

%% file: sesh-macros.tex
\newcommand{\ignore}[1]{}

\newcommand{\poly}{\mathrm{poly}}

\newcommand{\eqdef}{:=}

\newcommand{\Eqn}[1]{\hyperref[eq:#1]{(\ref*{eq:#1})}} 
\newcommand{\Fig}[1]{{Fig.\,\ref{fig:#1}}} 
\newcommand{\Tab}[1]{\hyperref[tab:#1]{Tab.\,\ref*{tab:#1}}} 
\newcommand{\Thm}[1]{\hyperref[thm:#1]{Theorem\,\ref*{thm:#1}}} 
\newcommand{\Fact}[1]{\hyperref[fact:#1]{Fact\,\ref*{fact:#1}}} 
\newcommand{\Lem}[1]{\hyperref[lem:#1]{Lemma\,\ref*{lem:#1}}} 
\newcommand{\Prop}[1]{\hyperref[prop:#1]{Prop.~\ref*{prop:#1}}} 
\newcommand{\Cor}[1]{\hyperref[cor:#1]{Corollary~\ref*{cor:#1}}} 
\newcommand{\Conj}[1]{\hyperref[conj:#1]{Conjecture~\ref*{conj:#1}}} 
\newcommand{\Def}[1]{\hyperref[def:#1]{Definition~\ref*{def:#1}}} 
\newcommand{\Alg}[1]{\hyperref[alg:#1]{Alg.~\ref*{alg:#1}}} 
\newcommand{\Clm}[1]{\hyperref[clm:#1]{Claim~\ref*{clm:#1}}} 
\newcommand{\Obs}[1]{\hyperref[obs:#1]{Observation~\ref*{obs:#1}}} 
\newcommand{\Rem}[1]{\hyperref[rem:#1]{Remark~\ref*{rem:#1}}} 
\newcommand{\Con}[1]{\hyperref[con:#1]{Construction~\ref*{con:#1}}} 
\newcommand{\Step}[1]{\hyperref[step:#1]{Step~\ref*{step:#1}}} 
\newcommand{\Assumption}[1]{\hyperref[assm:#1]{Assumption\,\ref*{assm:#1}}} 

\newcommand{\degen}{\kappa}
\newcommand{\tw}{t}
\newcommand{\dtw}{\tau}
\newcommand{\TRICONJ}{Triangle Detection Conjecture}
\newcommand{\LIPCO}{LIPCO}